\documentclass[12pt]{article}

\usepackage{amsmath, amsthm, amssymb}
\usepackage{geometry, graphicx}
\newtheorem{proposition}{Proposition}

\title{The Financialization of Proof-of-Stake: Asymptotic Centralization under Exogenous Risk Premiums}

\author{Mikhail Perepelitsa\footnote{
Department of Mathematics, University of Houston,
\emph{PGH, 3551 Cullen Blvd, 
Houston, TX 77204-3008, USA},
\texttt{maperepelitsa@uh.edu} }
}

\date{}

\begin{document}

\maketitle

\begin{abstract}
    This paper introduces a heterogeneous macroeconomic model of a Proof-of-Stake (PoS) network to analyze the long-term centralizing effects of external traditional finance (TradFi) yields. We model a continuum of rational actors divided into two distinct classes: investors, who optimize portfolios between staking and external variance-dominated investments, and consumers, who balance staking yields against the transactional utility of holding liquid assets. By employing a quasi-linear utility function to model consumer behavior, we derive a cubic polynomial that strictly defines the unique macroeconomic equilibrium of the coupled network. The model demonstrates that, at scale, external macroeconomic factors force the complete institutional capture of the PoS consensus layer.
    Because investors have access to external risk premiums, their wealth compounds exponentially, leading to massive capital inflows that crush the protocol's internal staking yield to effectively zero. We show that as the yield is crushed, consumer wealth becomes strictly upper-bounded. Ultimately, consumers are forced to cease staking entirely and hold all remaining wealth in liquid form to satisfy their transactional constraints.
\end{abstract}

\section{Introduction}

In recent years, there has been a growing interest among institutional investors in Ethereum (ETH), which, following its transition to a Proof-of-Stake (PoS) consensus mechanism, has established itself as a stable and energy-efficient alternative to Bitcoin. At present, it is estimated that the vast majority of all staked ETH is controlled by institutional entities and liquid staking providers, with Lido and top centralized exchanges historically accounting for over 50\% of the staking market (Nansen \cite{nansen2022merge}, Lido DAO \cite{lido2026goose}). 

The primary motivation of this research is to conduct a quantitative analysis of the effects that institutional investors exert on the underlying mechanics of the Ethereum network. These macroeconomic effects are not merely cosmetic; they are highly structural and potentially critical to the long-term decentralization and security of the PoS model.

It appears that the emerging philosophy of institutional investors toward ETH is to treat its staking contract as a low-risk, low-return asset. Rather than holding ETH for its transactional utility, institutions utilize the protocol's risk-free yield as a variance hedge against the highly volatile environment of mainstream finance. 

However, this financialization couples the internal Ethereum economy with the outside macroeconomic world, steering the protocol in an unintended direction. The core concern addressed in this paper is a mechanism of macroeconomic contagion: as investors flood the Ethereum staking pool to hedge outside variance, they severely depress the protocol's endogenous staking yield. This crushed yield subsequently pushes everyday network consumers---who must hold liquid ETH to pay for transactional friction---entirely out of the staking process. If staking becomes completely monopolized by institutional investors, this critical consensus layer becomes dangerously oversensitive to the volatility of outside traditional finance (TradFi) markets.

To formally analyze this scenario, we model a stylized, heterogeneous economy of ETH consisting of two distinct actors: \textit{investors} (portfolio rebalancers who maximize expected returns between ETH staking and an outside TradFi investment) and \textit{consumers} (network users who balance the transactional utility of holding liquid ETH against the opportunity cost of the staking yield).

By formalizing this heterogeneous economy, we derive the exact macroeconomic equilibrium of the network and project its long-term dynamics. The paper yields the following primary results.

In section \ref{sec:Investor} we consider the an idealized case in which the economy is dominated by the investors only. We identify three distinct macroeconomic scaling regimes for capital inflows into the PoS network, determined entirely by the external risk-adjusted opportunity cost ($\Delta = \mu_r - \sigma_r^2$). In the variance-dominated regime ($\Delta < 0$), staked capital scales linearly with total wealth ($S^* \propto W^1$), meaning the network functions as an effective variance hedge and absorbs external capital proportionally. At the critical phase boundary ($\Delta = 0$), staked capital exhibits fractional criticality, scaling sub-linearly with wealth ($S^* \propto W^{2/3}$). Finally, in the yield-dominated regime ($\Delta > 0$), staking completely decouples from wealth ($S^* \propto W^0$). In this state, the absolute mass of staked capital hits a definitive upper bound, reducing the staking sector to a fixed-size, niche allocation that cannot scale alongside broader institutional capital inflows.

In section \ref{sec:investor-consumer} we prove that the interaction between institutional portfolio optimization and retail transactional utility resolves into a single governed strictly by a cubic polynomial. 

In section \ref{sec:dynamics} we extend the model of a heterogeneous market into a discrete-time dynamic system, we prove that the compounding nature of outside TradFi yields makes institutional staking capture inevitable. As the system scales, institutional capital  crushes the endogenous yield to zero, forcing consumers to abandon the staking pool entirely ($S_c^t = 0$) and holding 100\% of their bounded wealth as liquid utility tokens.

To isolate the structural mechanics of the yield curve and capital flows, this model relies on a few clarifying assumptions. First, the fiat price of ETH is abstracted out of the analysis. We measure all wealth, yields, and utility strictly in native ETH units. We assume that while fiat price fluctuations alter the absolute dollar value of the network, rational actors operating within the ecosystem ultimately make staking and liquidity decisions based on native relative yields and native gas constraints. 

Second, for the core heterogeneous equilibrium and dynamic proofs, we assume baseline transaction fees ($F$) are zero (or strictly negligible compared to issuance). This is done for mathematical clarity. Adding expected transaction fees ($\mu_F$) to the protocol's yield does not alter the results or asymptotic limits of the model. Because fee revenue scales inversely with aggregate network stake ($1/S$) much like the issuance curve ($1/\sqrt{S}$), adding $F$ merely shifts the coefficients of the master cubic equation without changing its degree, its roots' stability, or the terminal condition of consumer staking extinction.

The internal macroeconomics of Ethereum, and Proof-of-Stake (PoS) blockchains in general, have been rigorously studied in a series of foundational papers. Saleh \cite{saleh2021blockchain} established the baseline economic viability of PoS consensus without the energy waste of Proof-of-Work, while subsequent literature has extensively explored token pricing, the economic mechanics of transaction fee markets \cite{roughgarden2021transaction}, and the effects of block rewards on equilibrium staking levels \cite{fanti2019economics,john2021equilibrium, cong2021tokenomics}. Furthermore, an ongoing debate regarding the inherent compounding effect of staking rewards \cite{leporati2023studying} and the long-term distribution of shares concluded that, under closed-system dynamics, staking does not inevitably lead to unchecked wealth concentration among the wealthiest validators \cite{rosu2021evolution}. 

Structurally, our model builds upon the recent macroeconomic and optimal issuance frameworks for PoS developed by Jermann \cite{jermann2025optimal, jermann2025macro}. However, our research diverges significantly from the existing literature by including the external macroeconomic forces. By introducing exogenous, variance-dominated traditional finance (TradFi) opportunity costs into the system, we demonstrate that the compounding nature of outside yields  overrides internal network dynamics. This guarantees the exponential wealth concentration and institutional capture that earlier closed-system models sought to rule out.

\section{The Homogeneous Model: The Pure Investor Economy}
\label{sec:Investor}
We consider an agent with initial wealth $W_0$ who allocates a fraction $w$ to Ethereum staking and the remaining fraction $(1-w)$ to an outside liquid option (e.g., DeFi).
Let $R_s$ be the discrete return of staking, with expected value $\mu_s$ and variance $\sigma_s^2$. \\
Let $R_r$ be the discrete return of the outside option, with expected value $\mu_r$ and variance $\sigma_r^2$. \\
Assume the two returns are uncorrelated ($\rho = 0$).

The total portfolio return $R_p$ is:
\begin{equation*}
    R_p = w R_s + (1-w) R_r
\end{equation*}

The wealth at the end of the period is $W_1 = W_0 (1 + R_p)$.
The agent seeks to maximize the expected change in logarithmic utility:
\begin{equation*}
    \max_w \mathbb{E} [ \ln(W_1) - \ln(W_0) ].
\end{equation*}
The maximum is given by the Kelly formula:







\begin{equation*}
    w = \frac{\mu_s - \mu_r + \sigma_r^2}{\sigma_s^2 + \sigma_r^2}
\end{equation*}
Since $w$ does not depend on the initial wealth of the agent $W_0,$ it is the same for all agents.
Thus, the total staked amount in the economy aggregates to $wW,$
where $W$ is the investors' total wealth.

We now connect this portfolio rule to the mechanics of the Ethereum blockchain.
Let $S$ be the total amount of ETH staked. The staking yield $\mu_s$ and staking variance $\sigma_s^2$ are driven by protocol issuance ($c$) and priority fees ($\mu_F, \sigma_F^2$):
\begin{equation*}
    \mu_s = \frac{c}{\sqrt{S}} + \frac{\mu_F}{S} \quad \text{and} \quad \sigma_s^2 = \frac{\sigma_F^2}{(S)^2},
\end{equation*}
so that the optimal Kelly fraction is  a function of $S:$
\begin{equation*}
    w(S) = \frac{\left( \frac{c}{\sqrt{S}} + \frac{\mu_F}{S} \right) - \mu_r + \sigma_r^2}{\frac{\sigma_F^2}{(S)^2} + \sigma_r^2}
\end{equation*}

In macroeconomic equilibrium, the total staked ETH ($S^*$) must equal the total system wealth ($W$) multiplied by agent's optimal fraction $w^*=w(S^*)$ at the level $S^*$ (clearance condition):
\begin{equation*}
    S^* = W w^*.
\end{equation*}

From this, we find the equation for the equilibrium:
\begin{equation*}
    \frac{S^*}{W} = \frac{\frac{c}{\sqrt{S^*}} + \frac{\mu_F}{S^*} - \mu_r + \sigma_r^2}{\frac{\sigma_F^2}{(S^*)^2} + \sigma_r^2}.
\end{equation*}

To convert this into a polynomial, we substitute $x = \sqrt{S^*}$ to get
\begin{equation}
\label{eq:quartic}
    (\sigma_r^2)x^4 + W(\mu_r - \sigma_r^2)x^2 - (W c)x + (\sigma_F^2 - W \mu_F) = 0.
\end{equation}

We restrict the analysis to the regime when total system wealth ($W$) is large enough that the expected fee revenue is larger than the fee variance: $W \mu_F > \sigma_F^2.$ This is a natural assumption in the current network environment.
With this assumption, the equation has a single positive root.
This follows from Descartes' rule of signs: if there is only one sign change in the coefficients of the polynomial, written from the highest order to the lowest, then the polynomial has exactly one real, positive root $x^*,$
which corresponds to the equilibrium level of staking $S^*=(x^*)^2.$

We assume that investors can not have a short position in the outside investment, that is that $w(S)\leq 1.$ This reduces to a lower bound on $S:$ $S\geq S_{min},$ for a certain $S_{min}.$ When the investors' wealth $W\leq S_{min},$ the economy is in the equilibrium $S^*=W,$ $w^*=1.$ If the system starts in this state, it will quickly get out of it, since the staking yield is diluted by the growing wealth.
In what follows we will assume that $W$ is above this staking threshold $S_{min}.$ Notice that for the equilibrium stake $S^*$ that we found above, $w(S^*)\geq0.$

\subsection{The asymptotic expansion and the phase transition}

The exact expression for the root of the quartic equation \eqref{eq:quartic} is not available.
Useful formulas can be derived in the limiting case when the total system wealth $W \to \infty$.
We demonstrate that the limiting behavior of the network bifurcates into two distinct regimes, determined entirely by the risk-adjusted premium of the exogenous market,  $\Delta=\mu_r - \sigma_r^2$.

\begin{proposition}
\label{prop:1}
As $W \to \infty$, the equilibrium staking amount $S^*$ exhibits three distinct asymptotic scaling regimes determined by the risk-adjusted premium $\Delta = \mu_r - \sigma_r^2$:
\begin{enumerate}
    \item If $\Delta < 0$ (Variance-Dominated External Regime), the absolute staked mass grows linearly with wealth ($S^* \propto W^1$), and the staked fraction converges to a strictly positive constant: 
    \[
    \frac{S^*}{W} \to 1 - \frac{\mu_r}{\sigma_r^2}.
    \]
    \item If $\Delta = 0$ (Critical Phase Boundary), the staked capital exhibits fractional criticality, scaling sub-linearly with wealth ($S^* \propto W^{2/3}$):
    \[
    (S^*)^{\frac{3}{2}} \sim \frac{c}{\sigma_r^2} W.
    \]
    \item If $\Delta > 0$ (Yield-Dominated External Regime), the absolute staked amount converges to a fixed constant independent of $W$ ($S^* \propto W^0$):
    \[
    \sqrt{S^*}\to \frac{c + \sqrt{c^2 + 4(\mu_r - \sigma_r^2)\mu_F}}{2(\mu_r - \sigma_r^2)}.
    \]
    Moreover, the staked fraction $w^* \to 0$.
\end{enumerate}
\end{proposition}

\begin{proof}
Consider the the variance-dominated external regime ($\Delta < 0$, i.e., $\sigma_r^2 > \mu_r$).
First, we divide \eqref{eq:quartic} by $W:$
\begin{equation} 
    \label{eq:scaled}
    \frac{\sigma_r^2}{W} x^4 + (\mu_r - \sigma_r^2) x^2 - c x - \mu_F + \frac{\sigma_F^2}{W} = 0,
\end{equation}

Assume $x$ scales as a power of $W$, such that $x \sim k W^\alpha$ for some constants $k > 0$ and $\alpha$.
In equation \eqref{eq:scaled}, the leading order terms with respect to $W$ are the $x^4$ term (growing as $W^{4\alpha - 1}$) and the $x^2$ term (growing as $W^{2\alpha}$).
To balance these dominant terms as $W \to \infty$, their exponents must be equal:
\begin{equation*}
    4\alpha - 1 = 2\alpha.
\end{equation*}
Substituting $x \approx k W^\frac{1}{2}$ into the dominant terms yields:
\begin{equation*}
    \sigma_r^2 k^4 - |\mu_r - \sigma_r^2| k^2 = 0
\end{equation*}
From this, we find that
\begin{equation*}
    k^2 = 1 - \frac{\mu_r}{\sigma_r^2}
\end{equation*}
Recalling that $S^* = x^2$, it follows that $S^* \sim k^2 W$.
Thus, the asymptotic staking fraction is exactly $1 - \frac{\mu_r}{ \sigma_r^2}$.

\vspace{0.5cm}
Consider the critical phase boundary ($\Delta = 0$, i.e., $\mu_r = \sigma_r^2$).
Equation \eqref{eq:quartic} when $\mu_r = \sigma_r^2$ becomes:
\begin{equation*}
    (\sigma_r^2)x^4 - (W c)x + (\sigma_F^2 - W \mu_F) = 0
\end{equation*}
We expect the term $Wcx$ to dominate the zero order term $\sigma_F^2 - W \mu_F.$ Under this condition we arrive at the equation
\begin{equation*}
    \sigma_r^2 x^4 - W c x \approx 0,
\end{equation*}
from which we find that $x^3 = \frac{c}{\sigma_r^2} W:$
\begin{equation}
    \label{eq:S critical}
    (S^*)^{\frac{3}{2}}{}\sim{}\frac{c}{\sigma_r^2} W.
\end{equation}
Since $x = (S^*)^{1/2},$ we obtain the claimed scaling $S^* \propto W^{2/3}$.

\vspace{0.5cm}
Consider the yield-dominated external regime ($\Delta > 0$, i.e., $\mu_r > \sigma_r^2$).
In this case, the only viable scaling is $\alpha = 0$, implying that  $x \sim \mathcal{O}(1)$.
As $W \to \infty$, the terms divided by $W$ in Equation \ref{eq:scaled} strictly vanish:
\begin{equation*}
    \lim_{W \to \infty} \left( \frac{\sigma_r^2}{W} x^4 \right) = 0 \quad \text{and} \quad \lim_{W \to \infty} \left( \frac{\sigma_F^2}{W} \right) = 0
\end{equation*}
The quartic equation degenerates into a limiting quadratic equation:
\begin{equation} 
\label{eq:limiting_quadratic}
    (\mu_r - \sigma_r^2) x^2 - c x - \mu_F = 0
\end{equation}
which has one positive root
\begin{equation*}
    x_\infty = \frac{c + \sqrt{c^2 + 4(\mu_r - \sigma_r^2)\mu_F}}{2(\mu_r - \sigma_r^2)}.
\end{equation*}
Therefore, $S^* \to x_\infty^2$, a fixed scalar constant. Because $S^*$ is bounded and $W \to \infty$, the equilibrium staking fraction $w^* \to 0$.
\end{proof}

\subsection{Sensitivity of staking to macroeconomic opportunity cost}

We analyze the responsiveness of the equilibrium staked capital, $S^*$, to shifts in the macroeconomic environment.
Let the risk-adjusted opportunity cost be defined as $\Delta \equiv \mu_r - \sigma_r^2$.
The equilibrium condition for the pure investor is given by the implicit function $F(S^*, \Delta) = 0$:
\begin{equation}
\label{eq:der 1}
    \sigma_r^2 (S^*)^2 + W \Delta S^* - W c \sqrt{S^*} - W \mu_F = 0
\end{equation}
where $W$ is total wealth, $c$ is the issuance parameter, and $\mu_F$ represents expected transaction fees.
To find the marginal impact of the opportunity cost on staked capital, we apply the implicit function theorem, computing $\frac{dS^*}{d\Delta}$.
Differentiating both sides of the equilibrium condition with respect to $\Delta$ yields:
\begin{equation*}
    2\sigma_r^2 S^* \frac{dS^*}{d\Delta} + W S^* + W \Delta \frac{dS^*}{d\Delta} - \frac{W c}{2\sqrt{S^*}} \frac{dS^*}{d\Delta} = 0
\end{equation*}
Solving for $\frac{dS^*}{d\Delta}$ we get
\begin{equation}
\label{eq:der 2}
    \frac{dS^*}{d\Delta} = -\frac{W S^*}{2\sigma_r^2 S^* + W \Delta - \frac{W c}{2\sqrt{S^*}}}
\end{equation}

Solving for $\Delta W$ from \eqref{eq:der 1} and using this expression in \eqref{eq:der 2}, we find that
\begin{equation}
\label{eq:der 3}
    \frac{dS^*}{d\Delta} = -\frac{W}{\sigma_r^2  +  \frac{W c}{2(S^*)^{\frac{3}{2}}}  +\frac{W \mu_F}{(S^*)^2} }<0.
\end{equation}
Thus, the system goes through the critical point continuously and monotonically.
Using the asymptotic expressions for the staking level $S^*$ obtained earlier, we find the corresponding expressions for the derivative $\frac{dS^*}{d\Delta}:$
\begin{itemize}
    \item When $\Delta < 0, $ 
    \[
    \frac{dS^*}{d\Delta} \propto \frac{1}{\sigma_r^2}W;
    \]
    \item When $\Delta =0,$
    \[
    \frac{dS^*}{d\Delta} \propto \frac{2}{3\sigma_r^2}W;
    \]
    \item When $\Delta >0,$
    \[
    \frac{dS^*}{d\Delta} \propto \frac{2(S^*)^2}{c\sqrt{S^*}+2\mu_F}.
    \]
\end{itemize}

Thus, in the variance-dominated and critical regimes, the sensitivity increases with the wealth, as a growing amount of capital moves in and out of the network. In the yield-dominated regime, the sensitivity asymptotically settles a certain value, independent of the wealth $W.$ 

\section{The Heterogeneous Model: Investors and Consumers}
\label{sec:investor-consumer}
We model a continuum of rational consumers and investors interacting within a Proof-of-Stake network.
Let $M$ represent the total circulating supply of the native asset.
We parameterize the macroeconomy as follows:
\begin{itemize}
    \item \textbf{Investors:} Control fraction $\alpha$ of the supply, yielding total wealth $W_i = \alpha M$.
They stake an optimal variance-hedging fraction $w$, resulting in an absolute staked capital of $S_i = w W_i$.
    \item \textbf{Consumers:} Control fraction $(1-\alpha)$ of the supply, yielding total wealth $W_c = (1 - \alpha)M$.
They optimize between staking (for yield) and holding liquid wealth ($L_c$) for transactional utility.
\end{itemize}

We model the consumer preferences over total wealth ($W_c$) and liquid asset holdings ($L_c$) with the utility function
\begin{equation*}
    U(W_c, L_c) = W_c + \gamma \ln(L_c),
\end{equation*}
where $\gamma$ is the aggregate transactional preference (measured in units of the native asset), which absorbs population size and base utility weights. 

That is, in the model with $N$ consumers, $\gamma =\gamma_0 N,$
where $\gamma_0$ is an individual consumer's preference, measured in ETH.
The consumer optimizes by setting their marginal rate of Substitution (MRS) between liquidity and wealth equal to the protocol's mathematical staking yield ($y$):
\begin{equation*}
    \text{MRS} = \frac{\partial U / \partial L_c}{\partial U / \partial W_c} = \frac{\gamma / L_c}{1} = \frac{\gamma}{L_c}=y.
\end{equation*}

\subsection{The equilibrium condition}
For simplicity of presentation, we will assume that the network does not pay fees ($F=0$).
In this case, the staking fraction of the investment portfolio is:
\begin{equation*}
    w(S){}={}\frac{\frac{c}{\sqrt{S}}- \Delta}{\sigma_r^2}
\end{equation*}
The network reaches equilibrium when the consumer's dimensionless MRS equals the protocol's global staking yield
and the investor staking demand based on staking yield matches the staking level that determines the yield.
Substituting $L_c = W_c - S_c$, and recognizing that the yield is diluted by both consumer stake ($S_c$) and investor stake ($S_i$), the equilibrium condition is:
\begin{equation*}
    \frac{\gamma}{W_c - S_c} = \frac{c}{\sqrt{S_i+ S_c}}
\end{equation*}
The staking clearance condition for the investors is:
\begin{equation*}
    w(S)W_i = S_i,\quad S=S_i+S_c.
\end{equation*}

To solve for the global macroeconomic equilibrium, we evaluate the system where both agents dynamically react to the aggregate yield.
Let $S$ be the total aggregate network stake ($S = S_i + S_c$).
We rewrite the two equilibrium conditions using the global protocol yield $y_t = \frac{c}{\sqrt{S}}$:
\begin{align*}
    \text{Investors: } & S_i = \frac{W_i}{\sigma_r^2} \left( \frac{c}{\sqrt{S}} - \Delta \right) \\
    \text{Consumers: } & \frac{\gamma}{W_c - S_c} = \frac{c}{\sqrt{S}}
\end{align*}

First, we solve the consumer's condition to isolate their absolute stake ($S_c$) as a function of the total network stake ($S$).
Multiplying both sides by $(W_c - S_c)$ and rearranging yields:
\begin{equation}
\label{eq:Sc_exact}
    S_c = W_c - \frac{\gamma}{c} \sqrt{S}
\end{equation}
This confirms that the consumer stakes their entire baseline wealth ($W_c$) strictly minus their liquidity requirement, which scales proportionally with $\sqrt{S}$.
Now, we write the total stake equal to the sum of both groups' individual stakes ($S = S_i + S_c$).
Substituting our behavioral equations into this identity:
\begin{equation*}
    S = \underbrace{\frac{W_i}{\sigma_r^2} \left( \frac{c}{\sqrt{S}} - \Delta \right)}_{\text{Investor Stake } (S_i)} + \underbrace{W_c - \frac{\gamma}{c} \sqrt{S}}_{\text{Consumer Stake } (S_c)}
\end{equation*}

Expanding the investor term and grouping the constant wealth components together:
\begin{equation*}
    S = \frac{W_i c}{\sigma_r^2 \sqrt{S}} - \frac{W_i \Delta}{\sigma_r^2} + W_c - \frac{\gamma}{c} \sqrt{S}
\end{equation*}

By substituting $x = \sqrt{S}$ and rearranging term we find the equation for the  equilibrium of the heterogeneous Ethereum economy:
\begin{equation}
\label{eq:cubic_master}
    x^3 + \left(\frac{\gamma}{c}\right) x^2 - \left(W_c - \frac{W_i \Delta}{\sigma_r^2}\right) x - \left(\frac{W_i c}{\sigma_r^2}\right) = 0
\end{equation}

By Descartes' Rule of Signs, this polynomial possesses exactly one unique, strictly positive real root $x^*=\sqrt{S^*}$, guaranteeing that the coupled heterogeneous network will always collapse into a single mathematically stable equilibrium.
Because the system is governed by a cubic polynomial, the exact algebraic root is cumbersome.

The remaining macroeconomic variables are computed from this value of $S^*:$
\begin{align*}
    \text{Liquid Consumer Wealth: } & L_c = \frac{\gamma}{c} \sqrt{S^*} \\
    \text{Consumer Stake: } & S_c = W_c - \frac{\gamma}{c} \sqrt{S^*} \\
    \text{Investor Stake: } & S_i = S^* - W_c + \frac{\gamma}{c} \sqrt{S^*}
\end{align*}

\subsection{Asymptotic limits}
To extract closed-form economic insights, we analyze the network under conditions of massive financialization, where total investor wealth scales to infinity ($W_i\to\infty$) and it grows faster than consumer wealth ($W_c/W_i\to0$). 
We restrict this analysis to the variance-dominated regime ($\Delta < 0$).
Returning to the master cubic equation:
\begin{equation*}
    S^{3/2} + \left(\frac{\gamma}{c}\right) S - \left(W_c - \frac{W_i \Delta}{\sigma_r^2}\right) \sqrt{S} - \left(\frac{W_i c}{\sigma_r^2}\right) = 0
\end{equation*}
Again, we try $S=kW_i^\alpha,$ for some positive constants $k$ and $\alpha.$
Balancing the  highest-order terms,
we find that $\alpha=1,$ and
\begin{equation}
    \label{eq:asymp_S}
    S \approx  - \frac{\Delta}{\sigma_r^2}W_i
\end{equation}
This shows that despite the non-linear inverse-square-root yield curve of the Ethereum protocol, at infinite scale, the total staked capital is linear in total wealth.

By substituting Equation \eqref{eq:asymp_S} back into our exact agent maps, we derive the closed-form asymptotic limits for the specific actors:


\textbf{1. Investor Stake ($S_i$):}
\begin{equation}
\label{het:S_i}
    S_i \approx  - \frac{\Delta}{\sigma_r^2}W_i  - W_c + \frac{\gamma}{c} \sqrt{S} \approx - \frac{\Delta}{\sigma_r^2} W_i
\end{equation}
At scale, the mathematical yield ($c/\sqrt{S}$) is crushed to effectively zero.
Therefore, investors completely ignore the protocol's internal yield and stake exactly their variance-hedging ratio (recall $\Delta < 0$), treating the network as a zero mean, zero risk investment. 

\textbf{2. Consumer Stake ($S_c$):}
\begin{equation}
\label{het:S_c}
    S_c = W_c - L_c \approx W_c - \frac{\gamma}{c} \sqrt{ - \frac{\Delta}{\sigma_r^2} W_i}.
\end{equation}
The evolution of $S_c$ depends how fast $W_c$ grows relative to $W_i.$ If the growth is slower (which is the case, as we will see in the next section) then at some moment $S_c=0.$
At this moment the economy separates into consumer holding a fixed amount of liquid ETH and not staking, and investors that do all staking .

\section{Long-Run Dynamics and Asymptotic Centralization}
\label{sec:dynamics}
In this section, we consider the dynamics of the ETH network starting with some initial wealth levels of investors ($W_i^0$) and consumers ($W_c^0$), assuming that the market instantaneously relaxes to the equilibrium described in the previous section.
The dynamics is defined recursively. Given $(W_i^t, W_c^t, S_i^t, S_c^t, y^t)$, we define the next period vector $(W_i^{t+1}, W_c^{t+1}, S_i^{t+1}, S_c^{t+1}, y^{t+1})$ in the following way.
New consumer wealth:
\begin{equation}
\label{dyn:W_c}
    W_c^{t+1} =  W_c^t + S_c^t y^t
\end{equation}

New wealth of investors:
\begin{equation}
\label{dyn:W_i}
    W_i^{t+1} =  S_i^t (1+y^t) + (W_i^t - S_i^t)(1+R_t)
\end{equation}
where $R_t$ is a random return with mean $\mu_r$ and variance $\sigma_r^2$.
Solving equation \eqref{eq:cubic_master} for the market $S^{t+1}$, we find:
\begin{equation}
    S_i^{t+1} = \frac{W^{t+1}_i}{\sigma_r^2}\left(\frac{c}{\sqrt{S^{t+1}}}-\Delta \right)
\end{equation}
\begin{equation}
    S_c^{t+1} = W_c^{t+1}-\frac{\gamma}{c}\sqrt{S^{t+1}}
\end{equation}
and the new yield:
\begin{equation}
\label{dyn:y^t}
    y^{t+1} = \frac{c}{\sqrt{S^{t+1}}}
\end{equation}

The following proposition shows that institutional investors eventually take over the staking of ETH.
\begin{proposition}
   Suppose that the outside market is governed by i.i.d. random variables $\{R_t\}$, with mean $\mu_r > 0$ and variance $\sigma_r^2$, and that it is variance-dominated ($\Delta = \mu_r - \sigma_r^2 < 0$).
   Given non-zero initial values of $W_i^0, W_c^0$, in the dynamic process \eqref{dyn:W_c}--\eqref{dyn:y^t}:
   \begin{itemize}
       \item the investor wealth grows exponentially $W_i^t \geq W_i^0 e^{g t}$, for some growth rate $g > 0$;
       \item there is a period $T$ after which the consumers stop staking ($S_c^t = 0$) and hold all their wealth in liquid ETH;
       \item the wealth of consumers ($W_c^t$) is upper bounded by some fixed amount $W_{max}$.
   \end{itemize}
\end{proposition}

\begin{proof}
We sketch an informal proof of this proposition. Consider the evolution of the investor wealth $W_i^t$ from equation \eqref{dyn:W_i}.
Using the asymptotic formula for $S_i^t$ from \eqref{het:S_i}, we can write:
\[
    W_i^{t+1} = W_i^{t}\left(1 + \left(1-\frac{\mu_r}{\sigma_r^2}\right)\frac{c}{\sqrt{S^t}} + \frac{\mu_r}{\sigma_r^2}R_t\right)
\]
We approximate this expression further with a deterministic compounding geometric growth rate $g^t$:
\[
    W_i^{t+1} = W_i^{t}e^{g^t}
\]
where
\[
    g^t = \mathbb{E}\left[\log\left( 1 + \left(1-\frac{\mu_r}{\sigma_r^2}\right)\frac{c}{\sqrt{S^t}} + \frac{\mu_r}{\sigma_r^2}R_t\right) \right] \approx \left(1-\frac{\mu_r}{\sigma_r^2}\right)\frac{c}{\sqrt{S^t}} + \frac{\mu_r^2}{2\sigma_r^2} \geq \frac{\mu_r^2}{2\sigma_r^2}
\]
In this way,
\begin{equation*}
   W_i^{t+1} \geq  W_i^{t}e^{\frac{\mu_r^2}{2\sigma_r^2}} 
\end{equation*}
which means that $W_i^t \geq W_i^0 e^{\frac{\mu_r^2}{2\sigma_r^2}t}$ and its growth rate is determined strictly by the outside environment.
Using this and \eqref{het:S_i}, we can estimate the yield:
\[
    y^t = \frac{c}{\sqrt{S_i^t+S_c^t}} \leq c\sqrt{\frac{\sigma_r^2}{-\Delta W_i^t}} \leq C_0 e^{-\frac{\mu_r^2}{4\sigma_r^2}t}
\]
with $C_0 = c\sqrt{\frac{\sigma_r^2}{-\Delta W_i^0}}$.
Next, from \eqref{dyn:W_c}, we find that:
\[
    W^{t+1}_c \leq W_c^t\left(1+ C_0 e^{-\frac{\mu_r^2}{4\sigma_r^2}t}\right)
\]
Thus, unrolling the recurrence yields an infinite product:
\[
 W_c^{t+1} \leq W_c^0 \prod_{k=1}^t \left(1+ C_0 e^{-\frac{\mu_r^2}{4\sigma_r^2}k}\right) \leq W_c^0 \prod_{k=1}^\infty \left(1+ C_0 e^{-\frac{\mu_r^2}{4\sigma_r^2}k}\right) = W_{max} < +\infty
 \]
 
Now, using the marginal utility balance, we find:
 \[
 \frac{\gamma}{W^t_c-S^t_c} = y^t \leq C_0 e^{-\frac{\mu_r^2}{4\sigma_r^2}t} \to 0
 \]
as $t$ increases.
Since the left-hand side is strictly bounded below by $\gamma/W_{max}$ (which is independent of $t$), there must be a finite period $T$ such that the interior solution breaks and $S_c^T = 0$.
That is, the consumers stop staking entirely and are permanently prevented from re-entering the staking pool for all $t > T$.
\end{proof}

\begin{table}[h!]
\centering
\renewcommand{\arraystretch}{1.2} 
\begin{tabular}{|l|c|r|}
\hline
\textbf{Parameter Description} & \textbf{Symbol} & \textbf{Simulation Value} \\ 
\hline\hline
Total Initial Token Supply & $M$ & $120,000,000$ \\ \hline
Initial Investor Wealth & $W_i^0$ & $12,000,000$ (10\%) \\ \hline
Initial Consumer Wealth & $W_c^0$ & $108,000,000$ (90\%) \\ \hline
Outside Expected Return & $\mu_r$ & $0.05$ \\ \hline
Outside Volatility & $\sigma_r$ & $0.30$ \\ \hline
Risk-Adjusted Opportunity Cost & $\Delta$ & $-0.04$ \\ \hline
Protocol Issuance Parameter & $c$ & $150$ \\ \hline
Aggregate Transactional Preference & $\gamma$ & $2,000,000$ \\ 
\hline
\end{tabular}
\caption{Baseline parameters for the numerical simulation of the heterogeneous dynamic economy, calibrated to approximate the Ethereum  and a variance-dominated external macroeconomic environment ($\sigma_r^2 > \mu_r$).}
\label{tab:sim_params}
\end{table}

\section{Numerical Simulation}
\label{sec:simulation}

To empirically illustrate the asymptotic limits derived in the previous section and demonstrate the finite-time dynamics of the heterogeneous economy, we simulate the recursive system under realistic macroeconomic and network parameters. While the theoretical proof guarantees that consumer staking eventually goes to zero in the limit ($t \to \infty$), the simulation demonstrates how rapidly this phase transition occurs under stochastic external market conditions.

At each discrete period $t$, the outside return $R_t$ is drawn from a normal distribution $\mathcal{N}(\mu_r, \sigma_r^2)$. The total stake is computed by finding the unique positive root of the master cubic equation \eqref{eq:cubic_master}, after which wealth is updated recursively according to equations \eqref{dyn:W_c} and \eqref{dyn:W_i}.

\subsection{Simulation results}

The results of the simulation confirm the structural capture of the consensus layer. As the simulation progresses, the investor wealth ($W_i^t$) compounds exponentially due to the external risk premium, while the consumer wealth ($W_c^t$) visibly stagnates, confirming the upper bound $W_{max}$ derived in our theoretical proof. 

This massive influx of compounding institutional wealth into the staking contract mathematically crushes the endogenous protocol yield ($y^t \to 0$). As the yield is compressed, the opportunity cost of holding liquid ETH vanishes. The rational consumers are forced to un-stake their assets to maintain the absolute liquidity ($L_c$) required by their transactional preference ($\gamma$). 

\begin{figure}[htbp]
    \centering
    \includegraphics[width=0.8\textwidth]{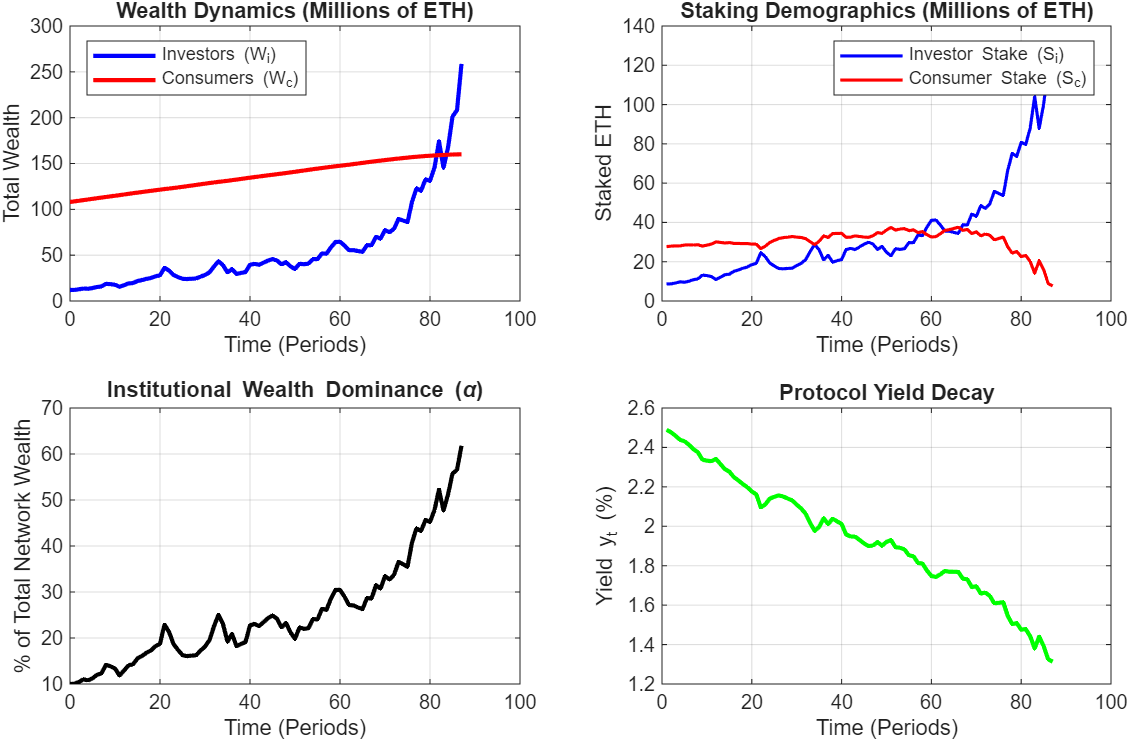} 
    \caption{Simulation of the heterogeneous economy over time. The top panels show the exponential divergence of investor wealth and the ultimate extinction of consumer staking. The bottom panels demonstrate the institutional takeover of total network wealth ($\alpha=W_i^t/(W_i^t+W_c^t)$) and the rapid decay of the protocol yield.}
    \label{fig:simulation}
\end{figure}

\section*{Acknowledgments}

During the preparation of this work, the author used Gemini 3.1 Pro (Google) to assist with LaTeX formatting, structuring the literature review, and refining the academic prose of specific sections. After using this tool, the author reviewed and edited the content as needed and takes full responsibility for the mathematical proofs, models, and final content of the publication.

\bibliographystyle{plain}
\bibliography{refs}

@unpublished{jermann2025macro,
  author = {Jermann, Urban},
  title  = {A Macro Finance Model for Proof-of-Stake {Ethereum}},
  note   = {Working Paper},
  year   = {2025}
}

@unpublished{jermann2025optimal,
  author = {Jermann, Urban},
  title  = {Optimal Issuance for Proof-of-Stake Blockchains},
  note   = {Working Paper, University of Pennsylvania},
  year   = {2025}
}

@article{saleh2021blockchain,
  author  = {Saleh, Fahad},
  title   = {Blockchain without Waste: Proof-of-Stake},
  journal = {The Review of Financial Studies},
  volume  = {34},
  number  = {3},
  pages   = {1156--1190},
  year    = {2021}
}

@article{rosu2021evolution,
  author  = {Ro\c{s}u, Ioanid and Saleh, Fahad},
  title   = {Evolution of Shares in a Proof-of-Stake Cryptocurrency},
  journal = {Management Science},
  volume  = {67},
  number  = {2},
  pages   = {661--672},
  year    = {2021},
  doi     = {10.1287/mnsc.2020.3791}
}

@unpublished{fanti2019economics,
  author = {Fanti, Giulia and Kogan, Leonid and Viswanath, Pramod},
  title  = {Economics of Proof-of-Stake Payment Systems},
  note   = {Working Paper},
  year   = {2019}
}

@article{cong2021tokenomics,
  author  = {Cong, Lin William and Li, Ye and Wang, Neng},
  title   = {Tokenomics: Dynamic Adoption and Valuation},
  journal = {The Review of Financial Studies},
  volume  = {34},
  number  = {3},
  pages   = {1105--1155},
  year    = {2021},
  doi     = {10.1093/rfs/hhaa089}
}

@unpublished{john2021equilibrium,
  author = {John, Kose and Rivera, Thomas J. and Saleh, Fahad},
  title  = {Equilibrium Staking Levels in a Proof-of-Stake Blockchain},
  note   = {Working Paper, SSRN 3965599},
  year   = {2021},
  url    = {https://ssrn.com/abstract=3965599}
}

@misc{nansen2022merge,
  author       = {{Nansen}},
  title        = {The Merge - A Deep Dive With Nansen},
  year         = {2022},
  url          = {https://research.nansen.ai/articles/the-merge-a-deep-dive-with-nansen},
  note         = {Accessed: 2026-04-26}
}

@misc{lido2026goose,
  author       = {{Lido DAO}},
  title        = {GOOSE-2025 \& EGGs-2025 Final Report},
  year         = {2026},
  url          = {https://research.lido.fi/t/goose-2025-eggs-2025-final-report/11304},
  note         = {Accessed: 2026-04-26}
}

@inproceedings{roughgarden2021transaction,
  author    = {Roughgarden, Tim},
  title     = {Transaction Fee Mechanism Design},
  booktitle = {Proceedings of the 22nd ACM Conference on Economics and Computation},
  year      = {2021}
}

@inproceedings{leporati2023studying,
  title     = {Studying the Compounding Effect: The Role of Proof-of-Stake Parameters on Wealth Distribution},
  author    = {Leporati, Alberto},
  booktitle = {Proceedings of the 5th Distributed Ledger Technology Workshop (DLT 2023)},
  series    = {CEUR Workshop Proceedings},
  volume    = {3460},
  pages     = {1--14},
  year      = {2023},
  url       = {https://ceur-ws.org/Vol-3460/papers/DLT_2023_paper_2.pdf}
}

\end{document}